\documentclass[a4paper,11pt]{article}

\usepackage{a4wide}

\usepackage{amssymb,amsmath,amsthm,url}
\usepackage{graphics,palatino}
\usepackage{tikz,pgffor}

\DeclareMathOperator{\C}{\mathit{C}\,}
\DeclareMathOperator{\CC}{\mathit{CC}\,}
\DeclareMathOperator{\CCC}{\mathit{CCC}\,}
\DeclareMathOperator{\CCCC}{\mathit{CCCC}\,}
\DeclareMathOperator{\K}{\mathit{K}\,}
\DeclareMathOperator{\KK}{\mathit{KK}\,}
\DeclareMathOperator{\KKK}{\mathit{KKK}\,}
\newcommand{\0}{\mathbf{0'}}
\DeclareMathOperator{\KH}{\mathit{K}^\0\,}
\DeclareMathOperator{\CH}{\mathit{C}^\0\,}
\DeclareMathOperator{\depth}{\mathit{depth}}

\newcommand{\la}{\langle}
\newcommand{\ra}{\rangle}
\newcommand{\rar}{\rightarrow}
\newcommand{\logg}{\log^{(2)}}
\newcommand{\loggg}{\log^{(3)}}
\newcommand{\logiv}{\log^{(4)}}
\newcommand{\logv}{\log^{(5)}}

\theoremstyle{plain}
\newtheorem{definition}{Definition}
\newtheorem{theorem}{Theorem}
\newtheorem*{theorem*}{Theorem}
\newtheorem{proposition}[theorem]{Proposition}
\newtheorem*{proposition*}{Proposition}
\newtheorem{lemma}[theorem]{Lemma}

\theoremstyle{remark}
\newtheorem{remark}{Remark}

\begin{document}

\title{Relating and contrasting plain and prefix Kolmogorov complexity}
\author{Bruno Bauwens\footnote{%
    LORIA, Universit\'e de Lorraine. I thank Paul Vitanyi and Alexander (Sasha) Shen for useful discussion
    (especially on  section~\ref{sec:relatingCandK}). 
    I thank Mathieu Hoyrup for encouragement to write down these results
    (and for arranging funding).
    }
}

\date{}

\maketitle  

\begin{abstract}
  In~\cite{BauwensCompcomp} a short proof is given that some strings have maximal plain Kolmogorov
  complexity but not maximal prefix-free complexity. 
  We argue that the proof technique is useful to simplify existing proofs and to solve open questions.

  We present a short proof of Solovay's result~\cite{Solovay} relating plain and prefix complexity:
  \begin{eqnarray*}
   \K(x) &=& \C(x) + \CC(x) + O(\CCC(x)) \\
   \C(x) &=& \K(x) - \KK(x) + O(\KKK(x)) \,,
  \end{eqnarray*}
  (here $\CC(x)$ denotes $\C(\C(x))$, etc.).

  We show that there exist $\omega$ such that 
  $\liminf \C(\omega_1\dots \omega_n) - \C(n)$ is infinite and
  $\liminf \K(\omega_1\dots \omega_n) - \K(n)$ is finite, i.e. the infinitely often 
  $C$-trivial reals are not the same as the infinitely often $K$-trivial reals
  (i.e.~\cite[Question~1]{BarmpaliasQuestions}). 
  
  We answer a question from Laurent
  Bienvenu: some $2$-random sequence have a family of initial segments with bounded plain
  deficiency (i.e. $|x| - \C(x)$ is bounded) and unbounded prefix deficiency 
  (i.e. $|x| + \K(|x|) - \K(x)$ is unbounded). 
  
  Finally, we show that
  there exists no monotone relation between probability and expectation bounded randomness
  deficiency, i.e. \cite[Question~1]{GacsTestsInClass}.
\end{abstract}

\section{Introduction}

Plain Kolmogorov complexity $\C(x)$ of a bitstring~$x$ was independently defined 
by Ray Solomonoff~\cite{solomonoffI} and later by 
Andrei Kolmogorov~\cite{Kolmogorov65} as the minimal length of a program that produces~$x$ on a
Turing machine.
In both definitions programs are strings of zeros and ones written on a work tape; the beginning and
end of the program is marked by blanc symbols.
During the execution, the Turing machine (which we call {\em plain} machine) can scan the beginning
and end of the program and use its length as additional information during the computation. 
After the computation, the output string should appear on the work tape, again the beginning and end 
should be marked by blank symbols (see~\cite{LiVitanyi,GacsNotes} for details). 
Kolmogorov complexity on such a machine is called  {\em plain} complexity. It is currently the most
popular notion of Kolmogorov complexity.

A closely related notion of complexity was introduced by
Leonid Levin~\cite{LevinPpi,LevinCK} and Gregory Chaitin~\cite{Chaitin75}
and has many applications in the study of algorithmic randomness.
Imagine a Turing machine on which programs are presented on a separate $2$-symbol input tape.
The tape does not have blanc symbols, only zeros and ones.
During the execution more input is scanned until the machine reaches a halting state, after which 
an output $x$ is defined.  We write $U(p) = x$ if $p$ is the minimal initial segment of the input
tape that contains all scanned cells and if the result of the computation is~$x$.
During the computation, the length of $p$ is no longer available. Programs on such a machine are
also called  {\em self-delimiting}. 
Note that the set of programs on which $U$ halts is prefix-free.
The minimal length of a program outputting $x$ on such a machine 
is called  {\em prefix} complexity $\K(x)$.

Prefix complexity is larger (up to an $O(1)$ constant) than plain complexity 
and the difference is at most $O(\log |x|)$, where $|x|$ denotes the length of $x$. 
For many applications this difference is not important.
However, for applications in the theory of algorithmic randomness, often $O(1)$-precise relations  
are used, and often one raises the question what happens when plain and prefix
complexity are exchanged in a result or a definition.
The goal of the paper is two-fold. First, we present a simple proof on a result that relates plain and prefix
complexity. Secondly, we refine a proof-technique (from~\cite{BauwensCompcomp}) 
to build strings where plain and prefix complexity 
behave differently, and apply it to solve three open questions. 
\bigskip

Several results are related to one of the
oldest questions in algorithmic randomness, raised by 
Robert Solovay~\cite{Solovay} (see~\cite[page 263]{Downey}). 
The maximal plain complexity of a string of length $n$ is $n + O(1)$ and we say that a string has 
$c$-maximal complexity if $\C(x) \ge |x|-c$. 
Martin-L\"of observed that for no $c$ and no infinite sequence all initial segments $x$ have 
$c$-maximal complexity.
On the other hand, the class of sequences for which some $c$ and infinitely many initial segments
$x$ exist with $\C(x) \ge n - c$ has measure one.
Similar observations hold for prefix complexity, (where the maximal complexity is $n + K(n) + O(1)$).
Solovay's question is whether
the classes of sequences with infinitely often maximal plain and prefix complexity are the
same; in other words, is $\liminf_{x \sqsubset \omega} |x| - \C(x)$ finite 
iff $\liminf_{x \sqsubset \omega} \K(|x|) + |x| - \K(x)$ is finite?


To answer this question, Solovay investigated whether there was a monotone relation 
between $\C(\cdot)$ and $\K(\cdot)$. He found that this was approximately the case
by showing
  \begin{eqnarray*}
   \K(x) &=& \C(x) + \CC(x) + O(\CCC(x)) \\
   \C(x) &=& \K(x) - \KK(x) + O(\KKK(x)) \,,
  \end{eqnarray*}
where complexity of a number $n$ is the complexity of the $n$-bit string $00\dots 0$ and
where $\CC(x)$, $\KK(x)$, etc, be short for $\C(\C(x))$, $\K(\K(x))$, etc. 
The proof in~\cite{Solovay} is cumbersome and Joseph Miller~\cite{MillerContrasting} 
made some simplifications using symmetry of information for prefix complexity. 
Here we use this technique to give an even much simpler proof. 
(Readers only interested in this result can directly go to
sections~\ref{sec:prerequisites} and \ref{sec:relatingCandK}.)

\smallskip
Solovay showed that the continuation of the first equation with terms up to $O(\CCCC(x))$ does
not hold. He also showed that maximal prefix complexity implies maximal plain complexity, 
but the reverse is not true:
there exist infinitely many $n$ and $x$ of length $n$ such that $n - \C(x) \le O(1)$ and
\begin{equation}\label{eq:prefixDef}
  \K(n) + n - \K(x) \ge \logg n - O(\loggg n) \,.
  \end{equation}
In~\cite{BauwensCompcomp} a simple proof (and generalizations) are presented.
Here we further develop the proof technique to solve several open questions.

Despite this negative result,
Miller~\cite{Miller2randC,Miller2randK} showed a positive answer for Solovay's question:
the sequences that have infinitely many initial segments with maximal plain and prefix complexity
are the same.
The proof is indirect: it shows that both classes coincide with the class of $2$-random sequences, 
i.e. Martin-L\"of random sequences relative to the halting problem (the equivalence of the first class
with $2$-randomness was also shown in~\cite{Nies2rand}).
Miller raised the question whether an (elegant) direct proof exists.
In~\cite{2rand} simple proofs of these equivalences with $2$-randomness are given, but still 
no direct proof. 
It is also shown that
\[
  \liminf_{x \sqsubset \omega} [|x| - \C(x)] = \liminf_{x \sqsubset \omega} [\K(|x|) + |x| - \K(x)] + O(1) \,,
  \]
by showing both sides equal $2$-randomness deficiency (see further). 
Laurent Bienvenu~\cite{personalBienvenuNov2012} asked whether for a $2$-random sequence, 
the initial segments for which plain and prefix-free 
complexity are maximal are the same; more precisely, for $2$-random $\omega$, does there exist $c$ and $d$ 
such that for all $n$: $n-\C(\omega_1\dots\omega_n) \le c$ implies $\K(n) + n - \K(\omega_1\dots\omega_n) \le d$?
(For some $c$ and $d$ the reverse implication is always true.)
We show that this is not the case: for every $3$-random sequence (a subset of the $2$-random
sequences) there are infinitely many initial segments $x$ with $|x| - \C(x) \le O(1)$ for which~\eqref{eq:prefixDef} holds.
This makes the existence of a simple direct proof unlikely.
We refer to  section~\ref{sec:contrastingIn3Random} for the proof of this result.


\bigskip
In algorithmic information theory, many relations are known between highly random sequences and
highly compressible sequences~\cite[Section 3.5]{BarmpaliasQuestions}. The second application of our technique 
considers one such class called {\em the infinitely often $K$-trivial sequences}: 
the sequences $\omega$ for which there exist $c$ and infinitely many $n$ such that 
$\K(\omega_1\dots\omega_n) \le \K(n) + c$, i.e. 
\[ 
  \liminf_n \left[ \K(\omega_1\dots \omega_n) - \K(n) \right] \le O(1)
  \] 
This class contains the computably enumerable sequences and the (weakly) $1$-generic sequences. 
Similar observations hold for the infinitely often $C$-trivial sequences, i.e. the sequences 
for which \[\liminf_n [\C(\omega_1\dots \omega_n) - \C(n)] \le O(1)\,. \]
Question 1 in~\cite{BarmpaliasQuestions} asks whether both classes coincide.
We show that this is not the case.

\bigskip
A last application of the proof technique concerns randomness deficiency for infinite sequences. 
Suppose one million zeros are prepended before a random string. The new string is still random, but 
one might argue that it is somehow ``less random''.
Randomness deficiency quantifies the amount of structure in a random sequence (see~\cite[Section 3.6.2]{LiVitanyi}
and~\cite{GacsTestsInClass}).
Let $\mu(\omega)$ denote the uniform measure.
Two closely related notions of deficiency exist in literature. 
\begin{itemize}
  \item 
    A lower semicomputable\footnote{
	  A non-negative rational function $f$ on $\{0,1\}^\infty$ is {\em basic} if
	  $f(\omega)$ is determined by a finite prefix of $\omega$.
	  A function $f$ into $\overline{\mathbb{R}}^+$, 
	  is {\em lower-semicomputable} if there exist a uniformly computable series of 
	  (non-negative) basic functions $f_i$ such that $f = \sum_i f_i$.
	  }
    function $f:\{0,1\}^\infty \rar \overline{\mathbb{R}}^+$ 
    (i.e. $\mathbb{R}^+$ extended with $+\infty$)  
    is a {\em probability bounded randomness test} if for each $k$ 
    \[
    \mu \left\{ \omega : f(\omega) \ge k \right\} \le k,\,.
    \]
  \item 
    A measurable function $f:\{0,1\}^\infty \rar \overline{\mathbb{R}}^+$ 
    is an {\em expectation bounded randomness test} if 
    \[
    \int_{\{0,1\}^{^\infty}} f(\omega) \text{d}\omega \le 1 \,.
    \]
\end{itemize}
The first notion is inspired by to the notion of confidence in statistical hypothesis testing, 
while the second is closely related, but mathematically more convenient to handle.
There exists a lower semicomputable
expectation bounded test $f_E$ that exceeds any other such test $g$ within a constant factor,
i.e. for all $g$ there exist $c$ such that $g \le cf_E$.  
The logarithm of such a universal test is called {\em expectation
bounded randomness deficiency~$d_E$}.
The deficiency depends on the choice of the universal test, but this choice affects the deficiency
by at most an additive constant.
Similar for probability bounded tests and  {\em probability bounded deficiency~$d_P$}.
Both deficiencies are related: $d_E = d_P + O(\log d_P)$,
and both deficiencies are finite iff the sequence is Martin-L\"of random.
We argue that the relation between plain and prefix complexity is very similar to 
the relation between $d_P$ and $d_E$.

Question 1 in  \cite{GacsTestsInClass} asks whether there exists a monotone
relation between probability bounded deficiency and expectation bounded deficiency that holds
within additive $O(1)$ terms. If this is not
the case then there exist two families of sequences $\omega_i$ and $\omega'_i$ such that
\[
 d_P(\omega_i) - d_P(\omega'_i) \rar +\infty
\]
for increasing $i$, and
\[
 d_E(\omega_i) - d_E(\omega'_i) \rar -\infty ,\,.
\]
In Section \ref{sec:contrastingDeficiencies}, we translate the main proof technique to 
deficiencies and construct such sequences. 
Hence, no monotone relation exists between the deficiencies.


\bigskip
The paper is organized as follows: first we discuss two old results which will be used throughout
the paper: Levin's formula relating plain and prefix complexity and Levin's formula for symmetry of
information.
In the next section we present a simple proof for Solovay's formulas relating $C$ and $K$.
All further results in the paper demonstrate different behaviour of $C$ and $K$ and the proofs have
a common structure. 
In  section~\ref{sec:contrastingCandK}, we repeat the simplest such proof by showing that
some strings have maximal plain but non-maximal prefix complexity.
Afterwards, in  section~\ref{sec:infOftenTrivial}, we show that the class of infinitely often $C$ and $K$ trivial sequences are different.
In  section~\ref{sec:contrastingIn3Random},  we show that each $3$-random sequence has infinitely many initial segments with maximal plain
complexity but non-maximal prefix complexity. Finally, in section~\ref{sec:contrastingDeficiencies}, we show that no monotone relationship exists
between plain and prefix randomness deficiency.
Section \ref{sec:relatingCandK}, sections \ref{sec:contrastingCandK}, \ref{sec:infOftenTrivial}, \ref{sec:contrastingIn3Random}, 
and section~\ref{sec:contrastingDeficiencies} can be red independently.

\section{Prerequisites}     \label{sec:prerequisites}

Two results are central in most proofs.
The first is Levin's symmetry of information~\cite{complexityOfComplexity}: for all $x,y$
\[
 \K(x) + \K(y|x,\K(x)) = \K(x,y) \,.
\]
The conditional variant is given by
\[
 \K(x|z) + \K(y|x,\K(x|z),z) = \K(x,y|z) \,.
\]

The second result relates plain and prefix complexity for random strings.
For all $n$-bit $x$:  $\C(x) = n + O(1)$ iff $\K(x|n) = n + O(1)$. 
We will use a more general variant.
\begin{lemma}[Folklore]\label{lem:relateDeficiencies}
  For all $j$ and $x$
  \[
   |j - \C(x)|  =  \Theta \left(\left| j - \K(x|j)\right| \right)
   \]
\end{lemma}

\begin{proof}
  The Lemma implies Levin's formula 
 \[
 \C(x) = \K(x|\C(x)) + O(1)\,,
 \]
 and in fact, it is equivalent to it:
 for any $j$ it implies $\K(x|j) = C(x)$ up to terms $O(\log |j - \C(x)|)$, and
 by the triangle inequality: 
 \[
   |j - \K(x|j)| = |j - \C(x)| + O\left(\log |j - \C(x)|\right)\,.\qedhere
   \]
\end{proof}

\section{Relating plain and prefix complexity}    \label{sec:relatingCandK}

Recall that $\KK(x)$, $\CC(x)$, etc, are short for $\K(\K(x))$, $\C(\C(x))$, etc.
\begin{theorem}\label{th:SolovayI}
  \begin{eqnarray}
   \K(x) &=& \C(x) + \CC(x) + O(\CCC(x)) \nonumber\\
   \C(x) &=& \K(x) - \KK(x) + O(\KKK(x)) \label{eq:second}\,.
  \end{eqnarray}
\end{theorem}

\begin{proof}
  Using symmetry of information we have
  \[
  \K(x) = \K(x,\K(x)) = \KK(x) + \K(x|\K(x), \KK(x) ) + O(1)\,.
  \]
  The last term equals $\K(x|\K(x)-\KK(x)\,) + O(\KKK(x))$.
  For $j = \K(x) - \KK(x)$ the equality is
  \[
   j = \K(x|j) + O\left(\KKK(x)\right)\,.
  \]
  Thus $\C(x) = j + O\left(\KKK(x)\right)$ by Lemma~\ref{lem:relateDeficiencies}, i.e. \eqref{eq:second}.
  
  \smallskip
  We obtain the first equation of the theorem from the second by showing that 
  \begin{eqnarray}
    \CC(x) &=& \KK(x) + O(\KKK(x)) \label{eq:CCvsKK} \\
    \KKK(x) &\le& O(\CCC(x)) \label{eq:CCCvsKKK}\,. 
  \end{eqnarray}
  For \eqref{eq:CCvsKK}, note that $a = b - c + O(d)$ implies $\C(a) = \C(b) + O(\K(c) + d)$. Applying
  this to  \eqref{eq:second} we obtain 
  \begin{equation*}
   \C(\C(x)) = \C(\K(x)) + O(\K(\KK(x)) + \KKK(x))\,.
 \end{equation*}
  Substituting $x \leftarrow \K(x)$ in \eqref{eq:second} gives
  \[
   \C(\K(x)) = \K(\K(x)) + \KK(\K(x)) + O(\KKK(x))\,.
  \]
  Combining both equations implies \eqref{eq:CCvsKK}.
  
  \smallskip
  It remains to show that \eqref{eq:CCvsKK} implies \eqref{eq:CCCvsKKK}. Using $\K(a) \le
  \K(b) + \K(b-a) + O(1)$:
  \[
   \K(\KK(x)) \le \K(\CC(x)) + \K(\KK(x) - \CC(x)) + O(1) 
  \]
  The first term at the right is bounded by $2\C(\CC(x))+O(1)$. 
  For the second, note that $\K(d) \le O(\log d)$ for any number $d$, 
  hence
  \begin{equation}\label{eq:relateKKKtoCCCprecise}
    \KKK(x) \le 2\CCC(x) + O(\log \KKK(x)) \,,
  \end{equation}
  i.e.  \eqref{eq:CCCvsKKK}.
\end{proof}

\begin{remark}\label{rem:CCCvsKKK}  The proof implies that 
$\K(x) = \C(x) + O(\CC(x))$ and $\KK(x) = \CC(x) + O(\CCC(x))$.
Alexander Shen raised the question whether $\KKK(x) = \CCC(x) + O(\CCCC(x))$? 
This does not hold. The proof is cumbersome and uses a topological argument
from~\cite{ShenTopological}, see appendix~\ref{sec:CCCvsKKK}.\footnote{
  \label{foot:DoubleComplexities}
  For later use in the appendix, note that the proof above also implies
  \[
  \CC(x), \; \mathit{CK}\,(x), \; \mathit{KC}\,(x), \; \KK(x),
  \]
  are all equal within error $O(\CCC(x))$ and error $O(\KKK(x))$. (Indeed, to relate $\KK(x)$ to
  $\mathit{KC}\,(x)$, apply $K(\cdot)$ to \eqref{eq:second}.)
  Moreover, for all $U,V,W,X,Y,Z \in \{C,K\}$ we have that
  $\mathit{UVW}\,(x) \le O\left(\mathit{XYZ}\,(x)\right)$.
  Indeed, by
  applying $\C(a) = \C(b) + O(\log (a-b))$ on the equalities above, 
  we obtain that $\mathit{CYZ}\,(x) = \mathit{CCC}\,(x) + O(\log \CCC(x))$.
  In the same way one shows that 
  $\mathit{KYZ}\,(x) = \mathit{KKK}\,(x) + O(\log \KKK(x))$.
  The result follows now from \eqref{eq:relateKKKtoCCCprecise}. 
  }
\end{remark}

\section{Contrasting maximal plain and prefix complexity}      \label{sec:contrastingCandK}

To get used to the main proof technique for the remainder of this paper, 
we start by showing the subsequent variant of Solovay's theorem.

\begin{theorem}[Solovay~\cite{Solovay}, Bauwens and Shen~\cite{BauwensCompcomp}]\label{th:SolovayII}
There exist infinitely many $x$ such that
$|x| - C(x) \le O(1)$ and $\K(|x|) + |x| - \K(x) \ge \logg |x|-O(1)$.
\end{theorem}

%
%

The main technique is to combine the two results from Section~\ref{sec:prerequisites}
with a third result: Peter G\'acs' quantification of incomputability 
of Kolmogorov complexity~\cite{complexityOfComplexity}.
He showed that for all lengths, there are $x$ 
such that $\K(\K(x)|x)$ is close to $\log |x|$ (and similar for plain complexity); 
if complexity were computable, then this would be bounded by $O(1)$.
The following tight variant from~\cite{BauwensCompcomp} will be used:
\begin{theorem}\label{th:GacsTight}
  For some $c$ and all $l$ there exist an $n$ such that $\log n = 2^l$, $\K(n) \ge (\log n)/2$ 
  and $\K(\K(n)|n) \ge l - c$. 
\end{theorem}
\begin{lemma}\label{lem:GacsTight}
  If $n$ satisfies the conditions of Theorem~\ref{th:GacsTight}, then
  \[
    \logg n = \log \K(n) + O(1) = \K(\K(n)|n) + O(1) \,.
  \]
\end{lemma}

\begin{proof}
Indeed, dropping additive $O(1)$ terms, the left equality follows from
\[
  \logg n \le \log ((\log n)/2) \le \log K(n) \le \log (2\log n) \le \logg n \,.
  \] 
It remains to show that $\K(\K(n)|n) \le \logg n$.
Indeed, $\K(\K(n)|n) \le \K(\K(n) | \logg n)$. 
and using $\logg n = \log \K(n)$ this follows from  $\K(i|\log i) \le \log i$.\footnote{
   \label{foot:Gacs} For the proof in the appendix note that this argument implies
   $\K(\K(n)|\logg n) = \logg n$. By Lemma~\ref{lem:relateDeficiencies} this implies 
   $\C(\K(n)) = \logg n$.
   }
\end{proof}

\bigskip

We informally explain why some strings have maximal plain complexity but non-maximal prefix
complexity. 
There exist plain machines $U$ for which a string $w$ exist such that $U(wx) = x$ for all $x$.
If $x$ has $O(1)$-maximal plain complexity, then $wx$ is an $O(1)$-shortest
program for $x$.
In a similar way, there exist a prefix machine $V$ 
such that for some $w$ we have $V(wx|\,|x|) = x$ for all $x$;
indeed, $V$ just copies the input from the program tape and uses 
the condition $|x|$ to know when to stop this operation.

If the length of $x$ is not available in the condition, 
no such trivial programs might exist.
To decide when to halt the copying procedure,
the length of $x$ must somehow be represented 
in the program in self-delimited form.
If the length of the program is minimal (within an $O(1)$ constant), 
this encryption of the length should also be minimal.
Mathematically, this corresponds to the following observations: $\K(x) = \K(n,x)$, (here and below we
omit $O(1)$ terms); and by symmetry of information
\[
 \K(n,x) = \K(n) + \K(x|n,\K(n)) \,.
\]
Thus, any shortest program for $x$ can be reorganized into a concatenation of two self-delimiting programs: 
the first computes $n$ and the second uses $n$ and the length of the first program to compute $x$.
The prefix deficiency is $\K(n) + n - \K(x) = n - \K(x|n,\K(n))$ and this is different from the
plain deficiency which is close to $n - \K(x|n)$ by  Lemma~\ref{lem:relateDeficiencies}. 
This explains why small prefix deficiency implies small plain deficiency, but not vice
versa.  In particular the deficiencies can only be different if $\K(\K(n)|n)$ is large, 
and this might indeed happen because of  Theorem~\ref{th:GacsTight}.

For appropriate $n$ the discussion explains how we construct $x$, it should contain $\K(n)$ 
and then be filled up further with bits independent from $n$ and $\K(n)$ until the plain complexity
is $n$. This is the approach in~\cite{BauwensCompcomp}, here we take advantage of the fact that
the program with largest computation time of length at most $n$ can also compute $\K(n)$ from $n$.
The proof below is even shorter than that of~\cite[Corrolary 6]{BauwensCompcomp}.

\begin{proof}
  As discussed above, we choose $n$, the length of $x$, such that 
  \begin{equation}\label{eq:compOfComp}
   \K(\K(n)|n) = \logg n +O(1)\,.
  \end{equation}
  By Theorem~\ref{th:GacsTight} and Lemma~\ref{lem:GacsTight}, there exist infinitely many such $n$. 
  Let $x = B(n)$ be the program of length at most~$n$ with maximal running time
  on a plain machine.  We drop $O(1)$ terms. Note that $\C(B(n)) = n = |B(n)|$.
  It remains to show $\K(B(n)) \le n + \K(n) - \logg n$ and this follows from
  \[
    \K(B(n)|n,\K(n)) \le n - \logg n
    \]
  (see above or note that $\K(B(n)) = \K(n,B(n)) = \K(n) + \K(B(n)|n,\K(n))$).
  From $n$ and $B(n)$ we can compute $\K(n)$, thus $n = \C(B(n)) = \K(B(n)|n)$ also equals 
  \[
      \K(\K(n),B(n)|n) \\
      = \K(\K(n)|n) + \K(B(n)|\K(n), \K(\K(n)|n), n)\,.
      \]
  Applying  \eqref{eq:compOfComp} twice implies 
  $n = \logg n + \K(B(n) | \K(n),n)$. 
\end{proof}

\begin{remark}
  As a corollary it follows that 
   $\K(x) = \C(x) + \CC(x) + \CCC(x) + O(\CCCC(x))$ is false.
  To show it contradicts Theorem~\ref{th:SolovayII} note that $\CCCC(x) \le O(\loggg(n)$.
  Let $x$ satisfy the conditions of the theorem and choose $y$ of length~$n$ with maximal plain and prefix
  complexity.  Now $\K(x) - \K(y) \ge \logg n - O(\loggg n)$.

  For similar reasons the following inequality is not an equality
  \[
   \K(x) \le \K(\C(x)) + \C(x)\,, \\
   \]
   see also Remark \ref{remark:openQuestion} below.
\end{remark}

\begin{remark}
  Miller generalized Solovay's theorem~\cite{MillerContrasting}. The proof above also 
  implies this generalization.
  
  \begin{theorem*}
    If a co-enumerable set (i.e. the complement can be algorithmically enumerated) of strings 
    contains a string of each length, then it also contains infinitely many strings $x$ such that 
    $K(|x|) + |x| - \K(x) \ge \logg |x| - O(1)$.
  \end{theorem*}

  This theorem also implies that the set of strings with maximal prefix complexity is not co-enumerable.
\end{remark}

  \begin{proof}
    Suppose $n$ satisfies the conditions of Theorem~\ref{th:GacsTight}. 
    Let $x$ be the lexicographically first string of length $n$ in the set. 
    We show that $x$ can be computed from $B(n+c)$ for some constant $c$, and this suffices because 
    we know from the proof above that $\K(BB(n+c)) \le n + \K(n) - \logg n + O(c)$.

    Consider a list of all strings of length $n$ and remove the strings outside the set using
    an enumeration of its complement. The moment the last string was removed can be computed with 
    a program of length $n + O(1)$ on a plain machine (by the total number of removed strings
    prepended with zeros to have an $n$-bit number). Thus, this moment must be before $B(n+c)$ for
    large $c$. 
  \end{proof}

\begin{remark}  
  The proof above can be used to contrast 
  {\em computational depth} with plain and prefix complexity. In~\cite[Tentative\footnote{
    Although it was called ``tentative'' definition, this version is simpler than the others 
    and is more often used in literature.
    }
  definition 1]{Bennett} the computational depth of a string $x$
  with precision $c$ is given by the minimal computation time 
  of a plain program for $x$ of length at most $\C(x) + c$:
  \[
  \depth_{C,c}(x) = \min \left\{t: |p| \le \C(x) + c \text{ and } U(p) = x \text{ in $t$ steps} \right\}\,.
  \]
  In a similar way, computational depth $\depth_{K,c}(x)$ with prefix machines can be
  defined.\footnote{ 
      We assume in all these definitions that the machine $U$ is  universal in the sense that for each other machine $V$ there 
      exist $w$ such that $U\left( wp \right) = V(p)$ each time $V(p)$ is defined and that simulating
      $V$ by $U$ in this way increases the computation time by a computable function.
      }
  With this assumption it follows easily that there exist a computable $f$ 
  such that $\depth_{K,c+2\log |x|}(x) \le f(\depth_{C,c}(x))$ and that $\depth_{C,c+2\log |x|}(x) \le f(\depth_{K,c}(x))$ 
  for $x$ of large length.
  The subsequent proposition shows that with higher precision, the equivalence is not possible. Let
  $BB(n)$ be the maximal computation time of a program of length at most $n$ on a plain machine
  (i.e. the computation time of $B(n)$).
\end{remark}

  \begin{proposition*}\label{prop:compareDepth}
    There exist a $c$ and infinitely many $x$ such that 
    $\depth_{C,c}(x)$ is bounded by a computable function of $x$ (and in fact bounded by a constant for an 
    appropriate universal machine) and $\depth_{K,\logg |x| - c}(x)$ 
    exceeds $BB(|x|-c)$.
  \end{proposition*}

\begin{proof}
  Consider the proof of Theorem~\ref{th:SolovayII}. 
  Rather than choosing $x$ to be $B(n)$, we fix some appropriate $c$ (see further), 
  and choose $x$ to be the lexicographically first $n$-bit string such that $\C(x) \ge n-2$ and 
  no self-delimiting program of length $n + \K(n)-c$ outputs $x$ in at most $BB(n)$ steps. 
  $x$ exist because for large $d$ there are at most $O(2^{n-d})$ 
  strings of length $n$ with complexity $n + \K(n)-d$ 
  (see~\cite[Theorem 3.7.6 p. 129]{Downey}, this also follows from the coding theorem).
  By construction $\C(x) \ge n - O(1)$ thus a trivial program of $x$ on a plain machine 
  is shortest within $O(1)$. Hence, the depth of $x$ is small on a plain machine. 
  Because $x$ can be computed from $B(n)$, the proof above guarantees that for infinitely many $n$ we have 
  $\K(x) \le \K(B(n)) + O(1) \le n + \K(n) - \logg n + O(1)$.  Fix such an $n$. To have
  $\depth_{K,\logg n - e}(x) < BB(n)$, we need a program for $x$ that computes $x$ in time less than 
  $BB(n)$ of length 
  $n + \K(n) - \logg n + O(1) + (\logg n - e) = n + \K(n) + O(1) - e$. For large $e$ this
  contradicts the choice of $x$, and hence the depth is at least $BB(n-O(1))$.
  \end{proof}

\begin{remark}\label{remark:openQuestion} 
  There exist infinitely many $x$ such that $\K(\K(x)|x,\C(x)) \ge \logg n - O(1)$. Indeed, 
  let $n$ be as in  Theorem~\ref{th:GacsTight}. Let $x$ of length $n$ have maximal prefix 
  (and hence plain) complexity such that $\K(\K(n)|x,n) \ge \K(\K(n)|n) - O(1)$. 
  This implies
  \[
    \K(\K(x)|x,\C(x)) = \K(n + \K(n)|x,n) = \K(\K(n)|x,n) \ge \K(\K(n)|n) \ge \logg n
    \]
    up to $O(1)$ terms.

    On the other hand $\K(\C(x)|x,\K(x))$ must be very small and it is an {\em open question} whether it is bounded by a
    constant. In particular this would imply that the inequality
    \[
      \K(x) \le \K(\C(x)) + \K\left(x|\C(x),\K(\C(x))\right)
    \]
    is an equality, which is also an  {\em open question}. 
\end{remark}

\section{Infinitely often $C$ and $K$ trivial sequences}      \label{sec:infOftenTrivial}

In the previous section we argued why a shortest self-delimiting program for a string
can contain more information than a shortest plain program. This suggest that the classes of 
infinitely often $C$ and $K$ trivial sequences might be different. The following theorem illustrates
this.

\begin{theorem}\label{th:trivialCvsK}
  There exists a sequence $\omega$ for which
  $\K(\omega_1\dots \omega_N) - \K(N) \le O(1)$ for infinitely many $N$, and for which
  $\C(\omega_1\dots \omega_N) - \C(N)$ tends to infinity.
\end{theorem}

\begin{figure}
 \begin{tikzpicture}
   \draw[thick,gray] node[anchor=east]  {\dots} (0,0)
      --  (7,0) node[anchor=west] {\dots};
      \draw[very thick,gray] (2,0) -- node[anchor=north,black] {$2^{n-1}$} ++(0,-0.1) ;
      \draw[very thick,gray] (4,0) -- node[anchor=north,black] {$2^{n}$} ++(0,-0.1) ;
      \draw[very thick,gray] (6,0) -- node[anchor=north,black] {$2^{n+1}$} ++(0,-0.1) ;

      \draw[ultra thick] (1.1,0) -- node[anchor=south] {$1w_{n-1}$} (1.95,0);
      \draw[ultra thick] (3.1,0) -- node[anchor=south] {$1w_{n}$} (3.95,0);
      \draw[ultra thick] (5.1,0) -- node[anchor=south] {$1w_{n+1}$} (5.95,0);
 \end{tikzpicture}
 \caption{Construction of $\omega$ in the proof of  Theorem~\ref{th:trivialCvsK}.}
 \label{fig:trivialCvsK}
\end{figure}
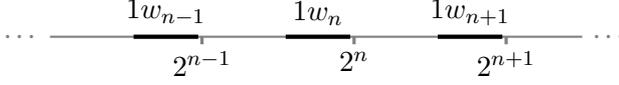

\begin{proof}
  Recall that $B(n)$ is a program of length at most $n$ with maximal running time on a plain
  machine.
 $\omega$ consists of zeros, except at small neighborhoods before indexes $2^n$ for all large $n$, 
 and in these neighborhoods strings $w_n = B(n + \logg n)$ are placed,
 see Figure~\ref{fig:trivialCvsK}; 
 more precisely $\omega_{2^n - |w_n|} \dots \omega_{2^n} = 1w_n$ (the prepended one in
 $1w_n$ allows us to identify the beginning of $w_n$).

 We show that $\C(\omega_1\dots \omega_N) - \C(N) \ge \loggg N - O(1)$ for all $N$, which 
 obviously tends to infinity.
 Fix any $N$ and let $n$ be such that $2^n \le N < 2^{n+1}$. 
 The initial segment $\omega_1\dots \omega_N$ computes $w_n$,
 thus $\C(\omega_1\dots \omega_N) \ge \C(w_n) \ge n + \logg n$ (here and below
 we omit terms $O(1)$).  On the other hand we have $\C(N) \le \log N = n$, hence 
 \[
   \C(\omega_1\dots \omega_N) - \C(N) \ge (n + \logg n) - n = \logg n = \loggg N\,.
 \]
 
 \smallskip
 It remains to construct $c$ and infinitely many $N$ such that $\K(\omega_1\dots \omega_N)\le \K(N)+c$. 
 The idea is to choose for infinitely many $n$ some $N$ such that
 $2^n \le N < 2^{n+1}-|w_{n+1}|$ and such that some shortest program for $N$ can compute $w_n$ with
 $O(1)$ of information; thus it can also compute $w_1, w_2, \dots, w_{n-1}$ and
 $\omega_1\dots\omega_N$ with $O(1)$ bits of information.

 As one might guess, we choose $n$ such that $\K(\K(n)|n) = \logg n$.
 Let us compute $\K(w_n| n,\K(n))$ in a similar way as before. We drop $O(1)$ terms: 
 \begin{eqnarray*}
   n + \logg n & = & \C(w_n) = \K(w_n|n) = \K(\K(n),w_n|n) \\
   &=& \K(\K(n)|n) + \K(w_n|\K(n),\K(\K(n)|n),n) \\
   &=& \logg n + \K(w_n|\K(n),n) \,. 
 \end{eqnarray*}
 Thus $\K(w_n| n,\K(n)) = n$. 
 
 Let $N$ in binary be the first $n-2$ bits of a program witnessing this equation
 (i.e. a program of length at most $n+O(1)$ computing $w_n$ from $n$ and $\K(n)$) prepended with
 the string ``$10$''. Prepending ``$10$'' guarantees that $2^n \le N < 2^{n+1} - |w_{n+1}|$ for large $n$. 
 By construction, if $n$ and $\K(n)$ are given, $N$ can compute $w_n$  with $O(1)$ bits of
 information. Thus it also computes $w_1, \dots, w_{n-1}$ and
 $\omega_1\dots\omega_N$. 
 On the other hand, every shortest program for $N$ can also compute $n$ and $\K(n)$
 with $O(1)$ bits of information. Indeed, 
 \[
 \K(N) = \K(N,n) = \K(n) + \K(N|n,\K(n))\,;
 \]
 thus on a universal prefix machine, there exist a $O(1)$-shortest program for $N$ 
 that is the concatenation of two self-delimiting programs 
 and the length of the first is $\K(n)$. 
 Together:
 \[
   \K(N) = \K(n,\K(n),N) = \K(w_1,\dots,w_n,n,\K(n),N)
   \ge \K(\omega_1\dots \omega_N)\,. \qedhere
   \] 
\end{proof}

\section{Contrasting plain and prefix complexity in $3$-random sequences}        \label{sec:contrastingIn3Random}

\begin{theorem}\label{th:unboundedK_boundedC}
  For every $3$-random sequence $\omega$ there are a $c$ and infinitely many $j$ such that
  $j - \C(\omega_1\dots \omega_j) \le c$ 
  and $\K(j) + j - \K(\omega_1\dots \omega_j) \ge \logg j -c$.
\end{theorem}

We conjecture that the result holds for all $2$-random sequences. It is possible to present the 
proof in game structure, but both the game and the strategy are quite complicated. We give
a proof that has the same core structure as the other proofs above.
In the proof we use two lemmas. The first roughly states that randomness deficiency of a string is bounded by 
the deficiency of an initial segment.
\begin{lemma}\label{lem:deficiencyInitialSegment}
  Let $j = |x|$ and $n = |xy|$
 \begin{eqnarray*}
  j - \K(x|j)  \le  n - \K(xy|j,n) + O(1) \\
 \end{eqnarray*}
\end{lemma}

\begin{proof}
  We omit $O(1)$ terms. Observe that $\K(xy|j,n) = \K(x,y|j,n)$, and this is bounded by
 \[
  \le  \K(x|j,n) + \K(y|x,j,n) \le  \K(x|j) + n-j \,, 
  \]
 because $\K(y | \,|y|) \le |y|$ for all strings $y$ and $|y|=n-j$ is computable from the condition.
 The inequality of the lemma follows after rearranging.
\end{proof}

Let $a$ and $b$ be two strings of the same length. Let $XOR(a,b)$ denote the bitwise XOR
operator on these strings. The following lemma states that if $a$ is incompressible, 
and $b$ is incompressible given $a$, then also $b$ is incompressible relative to $XOR(a,b)$. 
In fact, we will use a generalization which states
that if an extension $bw$ is incompressible given $a$, then this extension is incompressible 
given $XOR(a,b)$.

\begin{lemma}\label{lem:XORencryption}
  Let $a$ and $b$ be strings of equal length $\ell$, let $w$ be any string, let $n = |bw|$, and let $i$ be
  any number.  If 
  \[
    K(a|\ell,n,i) \ge \ell -c \text{\;\;\;   and \;\;\;  } K(bw|a,n,i) \ge n - c\,, 
  \]
  then
  \[
   \K(bw|XOR(a,b),n,i) \ge n - O(c)\,.
  \]
\end{lemma}

\begin{proof}
  In the lemma all complexities are conditional to $i$. 
  The proof of the conditional form follows the unconditional one, presented here.
  We first consider the case where $w$ is the empty string, 
  the proof for non-empty $w$ follows the same
  structure and will be presented afterwards. We need to show that for all $c,\ell,a,b$ such that 
  $|a|=|b|=\ell$, $\K(a|\ell) \ge \ell-c$ and $\K(b|a) \ge \ell -c$ we have
  \begin{equation}\label{eq:lemXorGoalSimple}
    \K(b|XOR(a,b)) \ge \ell + O(c) \,.
  \end{equation}
  
  \smallskip
  Indeed, 
  \[
   \K(a,b|\ell) = \K(a|\ell) + \K(b|a,\ell,\K(a|\ell)) + O(1)\,. 
  \]
  By assumption $\K(a|\ell) \ge \ell-c$, thus $\K(a|\ell) = \ell + O(c)$ and the last term simplifies to
  $\K(b|a,\ell) + O(c)$ and this equals $\ell + O(c)$. Hence $\K(a,b|\ell) = 2\ell + O(c)$.
  Let $xor = XOR(a,b)$. Because $a = XOR(b,xor)$ we have up to additive terms $O(c)$:
  \[
   2\ell  = \K(a,b|\ell) \le \K(xor, b|\ell) \le \K(xor|\ell) + \K(b|xor,\ell) \le \ell + \K(b|xor)\,,
  \]
  and this implies \eqref{eq:lemXorGoalSimple}.

  \smallskip
  We modify the equations above for the case where $w$ is not empty. Let $n = |bw|$ and remind that
  $|a|=\ell$. We start with 
  \[
   \K(a,b,w|\ell,n) = \K(a|\ell,n) + \K(b,w|a,\K(a|\ell,n),n) \ge \ell + n - O(c)\,. 
  \]
  Note that because $\ell=|b|$ we have $\K(bw,\dots|\ell,\dots) = \K(b,w,\dots|\ell,\dots)$.
   The left-hand also equals
  \[
   \K(xor, b,w|\ell,n) \le \K(xor|\ell,n) + \K(b,w|xor,\ell,n) \le \ell + \K(b,w|xor,n)\,,
  \]
  hence $\K(b,w|xor,n) \ge n - O(c)$.
\end{proof}

\begin{proof}[Proof of  Theorem~\ref{th:unboundedK_boundedC}]
  Let $\omega$ be $3$-random. 
  By  Lemma~\ref{lem:relateDeficiencies}, it suffices to construct infinitely many~$j$ such that 
 \begin{equation}\label{eq:condC}
  \K(\omega_1\dots \omega_j|j) \ge j - O(1)
 \end{equation}
 and $\K(\omega_1\dots \omega_j|j,\K(j)) \le j - \logg j + O(1)$. (Indeed, the last inequality implies
 $\K(\dots) \le j + \K(j) - \logg j + O(1)$ for the same reasons as in the proof of  Theorem~\ref{th:SolovayII}.)
 The second inequality follows from
 \begin{equation}\label{eq:condK}
   \K(\omega_1\dots \omega_{\logg j}|j,\K(j)) \le O(1)\,.
 \end{equation}

 \medskip
 \textbf{Sketch of the proof.}
 As usual, we construct $j$ such that $\K(\K(j)|j) \ge \logg j - O(1)$. For technical reasons, 
 we start with an index $i$ that will have almost the same information as $j$ and that
 satisfies  $\K(\K(i)|i) \ge \logg i - O(1)$.
 We also show that $i$ can be chosen such that $i$ and $K(i)$ are independent from $\omega_1\dots\omega_n$ 
 for an initial segment with maximal plain complexity (for this we need that $\omega$ is $3$-random).
 The main idea is to use $K(i)$ to encrypt the first $\log \K(i)$ bits of~$\omega$ (using the bitwise XOR operator). 
 Let $q$ be this encryption. We show (using Lemma~\ref{lem:XORencryption}) that 
 \[
   K(\omega_1\dots \omega_n|i,q,n) \ge n-O(1)\,.
   \]
 But with our encryption key $\K(i)$, we can decrypt the initial segment of $\omega$, thus 
 \[
 \K(\omega_1\dots \omega_{\log \K(i)}|i,q,\K(i)) \le O(1)\,.
 \]
 Finally, we define $j \le n$ by applying a bijective computable function of $i$ and~$q$. 
 Thus the pair $(q,i)$ contains the same information as $j$, i.e.
 $\K(\omega_1\dots \omega_n|i,q,n) = \K(\omega_1\dots \omega_n|j,n) + O(1)$. 
 Thus $\K(\omega_1\dots \omega_j|j,n) \ge j - O(1)$ by Lemma~\ref{lem:deficiencyInitialSegment}. 
 On the other hand, the construction implies that $\logg j = \log \K(i) + O(1)$ and that
 $\K(i)$ and $\K(j)$ carry the same information. Hence
 \begin{equation*}
  \K(\omega_1\dots \omega_{\logg j}|j,\K(j)) =  
  \K(\omega_1\dots \omega_{\log \K(i)}|i,q,\K(i)) + O(1)\le O(1)\,,
 \end{equation*}
 and this finishes the proof.


 \medskip
 \textbf{Requirements for $n,i$ and $q$.}
 We choose infinitely many triples $(n,i,q)$ and start with formulating five requirements 
 from which equations \eqref{eq:condC} and \eqref{eq:condK} follow. 
 Let $\la \cdot,\cdot \ra$ be a computable bijective pairing function from numbers and strings to
 numbers. For later use we assume that $\log \la k,x \ra = \log k + O(|x|)$ for all $k$ and $x$. 

 Equation \eqref{eq:condC} with $j = \la i,q \ra$, follows from
 \begin{itemize}
   \item[$(a)$] $\K(\omega_1\dots \omega_n|i,q,n) \ge n - O(1)$,
   \item[$(b)$] $\la i,q\ra \le n$ for large $n$, 
 \end{itemize}
  and Lemma~\ref{lem:deficiencyInitialSegment}. 
 Equation~\eqref{eq:condK} follows from:
 \begin{itemize}
   \item[$(A)$] $\K(\omega_1\dots \omega_{\log \K(i)} | \K(i),q) \le O(1)$,
   \item[$(B)$] $\log \K(i) = \logg \la i,q \ra + O(1)$,
   \item[$(C)$] $\K(i,q) = \K(i) + \log \K(i) + O(1)$.
 \end{itemize}
 Indeed, for all $z$, $(C)$ implies 
 $\K(z|i,q,\K(i)) = \K(z|i,q,\K(j)) + O(1)$.

 \medskip
 \textbf{Construction of $n$ and $i$.} 
 We use the characterization of $2$-random sequences with plain complexity:
 \begin{theorem*}[Joseph Miller~\cite{Miller2randC}, Nies--Stephan--Terwijn~\cite{Nies2rand}]
  A sequence $\omega$ is Martin-L\"of random relative to the Halting problem if and only if there exist a
  $c$ and infinitely many $n$ such that $\C(\omega_1\dots\omega_n) \ge n - c$.
 \end{theorem*}
 The proof of this theorem relativizes to the halting problem $\0$, i.e., a sequence is $3$-random 
 if and only if there are a $c$ and infinitely many $n$ such that 
 $\CH(\omega_1\dots\omega_n) \ge n - c$.
 Fix such an $n$. By Lemma~\ref{lem:relateDeficiencies}:
 \begin{equation}\label{eq:select_n}
   \KH(\omega_1\dots\omega_n| n)\ge n - O(1)\,.
 \end{equation}
 From now on we only use complexities that are conditional to~$n$. For notational
 simplicity we drop $n$ from the condition, thus $\K(a) \equiv \K(a|n)$, $\K(a|b) \equiv
 \K(a|b,n)$, etc.  
 
 Let $i$ be the largest number such that 
 \begin{itemize}
   \item[$(i)$] 
     $\K(\K(i)|i) \ge \logg i - c$ and $\K(i) \ge (\log i)/2$,
     where $c$ is the constant from Theorem~\ref{th:GacsTight}.
   \item[$(ii)$] 
     $\la i,x \ra \le n$ for all $x$ of length at most $1 + \logg i$.
 \end{itemize}
 Such $i$ exists because also the conditional version of Theorem~\ref{th:GacsTight} holds.
 In fact, for increasing choices of $n$, we find infinitely many such $i$.
 By  Lemma~\ref{lem:GacsTight}, the first condition implies 
 \begin{equation}\label{eq:logKiVsLoggi}
   \log \K(i) = \logg i + O(1)\,.
 \end{equation}
 Note that $i$ and $\K(i)$ can be computed from $\0$ and $n$, hence  \eqref{eq:select_n} implies
 \begin{equation}\label{eq:select_i}
  \K(\omega_1\dots\omega_n| i,\K(i)) \ge  n - O(1)\,.
 \end{equation}

 \medskip
 \textbf{Construction of $q$.}
 $q$ is given by the bitwise XOR-function of $K(i)$ in binary, 
 and the initial segment of $\omega$ with the same length:
 \[
 q = XOR\left(\omega_1\dots \omega_{\log K(i)}, \la K(i) \ra \right) \,.
 \]
 Because $XOR(a,XOR(a,b)) = b$ this implies $(A)$. 

 Recall that all complexities implicitly have $n$ in the condition and that $\K(\K(i)|i) \ge \logg i + O(1)$.
 Together with \eqref{eq:select_i}, this can be applied to Lemma~\ref{lem:XORencryption}  (with $l = \log \K(i) = \logg i + O(1)$,
 $bw = \omega_1\dots \omega_n$ and $a = \la \K(i) \ra$) 
 and we conclude that $\K(\omega_1\dots \omega_n|i,q,n) \ge n - O(1)$, i.e. condition $(a)$.

 For large $n$, we have large $i$, and hence $|q| = \log \K(i) \le \log (2\log i) = 1 + \logg i$. 
 By choice of $i$ (the second condition) this implies $(b)$. 
 We assumed that the pairing function satisfies $\log \la i,q \ra = \log i + O(|q|) = \log i +
 O(\logg i)$.  Thus $\logg \la i,q \ra = \logg i + O(1)$. By~\eqref{eq:logKiVsLoggi} this implies $(B)$.

 It remains to show $(C)$. 
 Note that
 \[
 \K(i,q) = \K(i) + \K\left(q | i, \K(i)\right)\,.
 \]
 The last term equals 
 $\K(\omega_1\dots \omega_{\log \K(i)}|i, \K(i))$.  By  \eqref{eq:select_i} and
 Lemma~\ref{lem:deficiencyInitialSegment}  this is at least
 $\log \K(i) + O(1)$, and in fact it is equal to this, because $K(z|\,|z|) \le |z|$ for all $z$. 
\end{proof}

\section{Contrasting expectation and probabilistically bounded deficiency}        \label{sec:contrastingDeficiencies}

Recall from the introduction that there exist two different notions 
of randomness deficiency for a sequence $\omega$. 
We start by showing that the two notions are related. 

\begin{proposition}\label{prop:characterize_d_P}
  \[
  d_P(\omega) = \sup \{k: d_E(\omega|k) \ge k\} + O(1)\;\;\;
  \footnote{
    Conditional probability bounded deficiency is defined in the natural way: it 
    is the logarithm of a multiplicatively maximal function $f(\cdot|k)$ that is 
    lower semicomputable uniformly in $k$, such that for each $k$ the function is a 
    probability bounded test.
    }
  \]
\end{proposition}

This characterization is closely related to a characterization of plain complexity in terms of prefix complexity 
(see~\cite[Lemma 3.1.1 p. 203]{LiVitanyi}):
\[
\C(x) = \min \left\{k: \K(x|k) \le k  \right\} + O(1)\,. 
\]
Many results relating and contrasting prefix and plain complexity on one side, 
can be translated to results about expectation and probability bounded
deficiency. (In these results $d_E(\cdot)$ corresponds to $\K(\cdot)$ and $d_P(\cdot)$ to $\C(\cdot)$.)

\begin{proof}
  For the $\ge$-direction 
  we need to show that the exponent of the supremum defines a lower-semicomputable probability bounded test. 
  $d_E$ is lower semicomputable, thus also the supremum is lower semicomputable, 
  and it remains to show that the measure where it exceeds $\ell$ is bounded by $O(2^{-\ell})$.
  By definition we have $\int 2^{d_E(\omega|k)} \text{d}\omega \le 1$ for all $k$, thus the measure of 
  $\omega$ such that $d_E(\omega|k) \ge k$ is at most $2^{-k}$.
  If the supremum exceeds $\ell$ for some $\omega$, then 
  $d_E(\omega|k) \ge k$ for some $k \ge \ell$. The total measure for which this can happen is 
  at most $2^{-\ell} + 2^{-\ell-1} + \dots \le O(2^{-\ell})$. 

  For the $\le$-direction note that every probability bounded test $f$
  defines a family of expectation bounded tests $g(\cdot|k)$ such that 
  $g(\omega|k) = 2^k$ iff $f(\omega) \ge 2^k$. 
  Indeed the condition implies $\int f(\omega|k) \text{d}\omega \le 2^k\cdot 2^{-k} = 1$. 
  Obviously, if $f$ is lower semicomputable, the tests $g(\cdot|k)$ 
  are lower semicomputable uniformly in $k$. 
  If $f$ is the universal test corresponding to $d_P$, then $d_P(\omega) \ge k$ implies
  $f(\omega) \ge 2^k$, which implies $g(\omega|k) \ge 2^k$ thus $d_E(\omega|k) \ge k - O(1)$.
\end{proof}

The question was raised in~\cite[Question 1]{GacsTestsInClass} whether the two deficiencies
are related by a monotone function, or
\textit{does there exist two families of sequences $\omega^{\ell}$ and $\omega'^{\ell}$ such that
\[
d_A(\omega^{\ell}) - d_A(\omega'^{\ell}) \rightarrow \infty
\]
for $\ell \rightarrow \infty$ and
\[
d_P(\omega^{\ell}) - d_P(\omega'^{\ell}) \rightarrow -\infty\,.
\]
} \!We show this is indeed the case.

\begin{theorem}\label{th:averageVsProbDeficiency}
  There exist families of sequences $\omega^{\ell}$ and $\omega'^\ell$ such that for infinitely many~$\ell$
  \[
   |d_P(\omega^{\ell}) - d_P(\omega'^\ell)| \le O(1)
   \] 
   if $\ell \rightarrow \infty$ and 
  \[ 
  d_E(\omega^\ell) - d_E(\omega'^\ell) \ge \ell - O(1)\,.
  \]
\end{theorem}
The positive answer to the question above follows by 
prepending $\ell/2$ zeros to $\omega'^\ell$ for all $\ell$. 
This decreases the complexities in the definition of $d_P(\omega'^\ell)$ and $d_E(\omega'^\ell)$  by
$\ell/2 + O(\log \ell)$ 
and hence increases these deficiencies by the same amount; and this is enough for the question.

Before presenting the proof, we show two lemmas that play the same role as symmetry of information 
and Levin's result relating plain and prefix complexity (i.e. Lemma~\ref{lem:relateDeficiencies}). 

\begin{lemma}[Symmetry of deficiency]\label{lem:symmetryOfDeficiency}
  For all $\omega$ and all $x$ that belong to a prefix-free computably enumerable set, we have 
  \[
    d_E(x\omega) = |x| - \K(x) + d_E(\omega|x,\K(x)) + O(1)\,,
  \]
  here $x\omega$ denotes concatenation of $x$ and $\omega$. 
  The $O(1)$-term depends on the choice of the computably enumerable set.
\end{lemma}

The proof uses a characterization of expectation bounded deficiency 
in terms of prefix Kolmogorov complexity (see for example~\cite[Proposition 2.22]{GacsTestsInClass}):
\begin{theorem*}
$ d_E(\omega|z)   = \sup_n \left\{n - \K(\omega_1\dots\omega_n|z)\right\} + O(1) $
\end{theorem*}

\begin{proof}[Proof of  Lemma~\ref{lem:symmetryOfDeficiency}.]
 Let $x$ be a member of the prefix-free computably enumerable set. From $xy$ we can 
 compute $x$ by enumerating the prefix-free set until an initial segment of $xy$ 
 and this segment can only be $x$.  Symmetry of information implies
 \[
  \K(xy) = \K(x,y) + O(1) = \K(x) + \K(y|x,\K(x)) + O(1)\,,
 \]
 i.e.
 \[
  |xy| - \K(xy) = |x|-\K(x) + |y| - \K(y|x,\K(x))\,.
 \]
 If we take on both sides the supremum of $y$ over all prefixes of $\omega$, we \textit{almost} 
 obtain the equation of the lemma; the problem is that in the definition of $d_E(x\omega)$ 
 we also need to consider prefixes $z$ of $x$. It remains to verify that 
 \[
   |z| - \K(z) \le |x| - \K(x) + O(1)
   \]
 for all prefixes $z$ of $x$. 
 In general this is false, but for $x$ in a prefix-free enumerable set it holds.
 For any $z$ and $x$,
 let $P(x|z) = 2^{-|x|+|z|}$ if $x$ is an extension of $z$ that belongs to the prefix-free set, 
 otherwise let $P(x|z) = 0$.
 Note that $\sum_x P(x|z) \le 1$ and $P(x|z)$ is lower-semicomputable, hence
 the coding theorem implies $\K(x|z) \le -\log P(x|z) + O(1) \le |x|-|z| + O(1)$. 
 Symmetry of information implies
 \[
  K(x) \le \K(x,z) \le \K(z) + \K(x|z) + O(1) \le \K(z) + |x| - |z| + O(1)\,,	
 \]
 and this implies the equation above.
\end{proof}

The analogue of Lemma~\ref{lem:relateDeficiencies} for deficiencies of sequences is
\begin{lemma}\label{lem:relateInfDeficiency}
  For all $j$ and $\omega$
  \[
    \left|j - d_E(\omega|j) \right| = \Theta \left|j - d_P(\omega)\right| \,.
  \]
\end{lemma}

\begin{proof}
 For fixed random $\omega$, the map $t \rightarrow d_E(\omega|t)$ maps points at distance $d$ 
 to points at distance $O(\log d)$. Hence, 
 the map has a unique fixed point $t$ within precision $O(1)$, i.e. $d_E(\omega|t) = t + O(1)$ for
 some $t$. 
 This implies that $t$ is $O(1)$-close to the minimal $s$ such that $d_E(\omega|s) \ge s$,
 i.e. $d_P(\omega)$.
 Our observation implies that $d_E(\omega|t+d) = t + O(\log d)$,
 thus for $j = t + d$ we have $j - d_E(\omega|j) = j - d_P(\omega) + O(\log (j-d_P(\omega)))$,
 and this implies the lemma.
\end{proof}

\begin{proof}[Proof of Theorem~\ref{th:averageVsProbDeficiency}.]
  For each $\ell$ we choose a $k$ such that $\logg k \le \ell$ and
  $\K(\K(k)|k) \ge \logg k - c$  where $c$ is the constant 
  from  Theorem~\ref{th:GacsTight}.
  By  Lemma~\ref{lem:GacsTight} 
  \begin{equation}\label{eq:annoying}
    \ell = \logg k = \log \K(k) + O(1)  \,.
  \end{equation}

  We choose $\omega$ such that 
  \[d_P(\omega|k,\K(k)) \le O(1)\,.\] 
  Let $0^k1\omega$ be the sequence that starts with $k$ zeros, followed by a one
  and followed by $\omega$. Let $0^k1\la\K(k)\ra\omega$ be $0^k1$ followed
  by $\K(k)$ in binary, followed by~$\omega$.
  The theorem follows from the values of the expectation and probability bounded deficiencies of
  these strings, given in the table below:
  \[
  \begin{array}{r|l|l}
     \alpha & d_E(\alpha) & d_P(\alpha) \\
     \hline
     0^k1\omega & k - \K(k) + O(1) & k + O(1) \\
     0^k1\la\K(k)\ra^l\omega & k - \K(k) + \ell + O(1)  & k + O(1)
   \end{array}
   \]
   It remains to prove that the values in the table are correct.

  \medskip
  The values of $d_E(\cdot)$ in the first column are obtained from Lemma~\ref{lem:symmetryOfDeficiency}. 
  In the first case, the prefix-free set is the set of strings $0^m1$ for all~$m$, thus
  \[
   d_E(0^k1\omega) = k - \K(k) + d_E(\omega|k,\K(k)) + O(1)\,. 
  \]
  In the second case, the prefix-free set is the of all strings $0^m1z$ for all~$m$ and all $z$ of
  length $\logg m$. Recall that $\K(k,\K(k)) = \K(k) + O(1)$, thus
  \[
   d_E(0^k1\la\K(k)\ra\omega) = k + \logg k - \K(k) + d_E(\omega|k,\K(k)) + O(1)\,. 
  \]
  
  \medskip
  To evaluate $d_P(\cdot)$ we use Lemma~\ref{lem:relateInfDeficiency}. Hence, let us compute
  $d_E(0^k1\omega|k)$. Again we use  Lemma~\ref{lem:symmetryOfDeficiency}:
  \[
   d_E(0^k1\omega|k) = k - \K(0^k1|k) + d_E(\omega|\K(0^k1|k),k) + O(1) = k + d_E(\omega|k) + O(1)\,. 
  \]
  This implies $d_P(0^k1\omega) = k + O(1)$.
  For the second case, note that $\K(0^k1\la\K(k)\ra|k) = \K(\K(k)|k) + O(1) = \logg k + O(1)$ by choice of $k$. 
  With similar reasoning we determine $d_P(0^k1\la \K(k) \ra\omega)$:
  \[
   d_E(0^k1\la\K(k)\ra\omega|k) = (k + \logg k) - \K(\K(k)|k) + d_E(\omega|\K(k), k) + O(1) \,.\\
   \]
   This equals $k + O(1)$ by  \eqref{eq:annoying}.
\end{proof}

\bibliography{eigen,kolmogorov}

\appendix

\section{$\KKK(x) = \CCC(x) + O(\CCCC(x))$ does not hold}   \label{sec:CCCvsKKK}

\begin{proposition}\label{prop:CCCvsKKK}
  There exist infinitely many $x$ such that $CCCC\,(x) \le O(\logv |x|)$ and
  \[
 \left|\CCC(x) - \KKK(x)\right| \ge \Omega(\logiv |x|)\,.
  \]
\end{proposition}

We use our main technique to contrast $\C(\cdot)$ and $\K(\cdot)$ 
to disprove that $\KKK(x) = \CCC(x) + O(CCCC\,(x))$. 
However, we will use the variant presented in the proof of~\cite[Corollary 6]{BauwensCompcomp}.
It is combined with a topological argument which was inspired by~\cite{ShenTopological}. 
We start with a definition and a lemma.

\begin{definition}
A set $S$ of numbers is $c$-dense in a superset $A$ if for each $a \in A$ 
there is an $s \in S$ such that $|a-s| \le c$.
\end{definition}

\begin{lemma}\label{lem:dense}
  If $S$ is $c$-dense in an interval of size $k$, then the set
  \[
    \{\K(k): k \in S\} 
  \]
  is $O(\log c)$-dense in some interval of size $\Omega(\log k - \log c)$.
\end{lemma}

\begin{proof}[Proof of  Proposition~\ref{prop:CCCvsKKK}.]
  Let $T$ be the set defined by the lemma. Note that the function $\K(\cdot)$ maps points at 
  distance $d$ to points at distance $O(\log d)$, hence $T$ is $O(\log c)$-dense in $[\min T, \max T]$.
  It remains to show that the maximum of this set differs from its minimum 
  by at least $\log k - O(\log c + \logg k)$. Let $r$ be the minimal number in the interval of size $k$
  (in which $S$ is dense) that ends with $\log k - 2$ zeros. 
  By assumption $r$ is at $c$ distance of an element in $S$. 
  On the other side, if the $\log k - 2$ last zeros of $r$ are changed, the
  corresponding number remains always in the interval of size $k$, and for one such change the complexity
  of $r$ must increase by at least $\log k - O(\logg k)$. (Otherwise, to many short descriptions
  exist of such modified $r$ and we could use this to obtain a shorter description for $r$.)
  This element is $c$-close to an element in~$S$, thus the difference of the minimum and the maximum 
  of $\K(k)$ over $S$ is at least $\log k - O(\log c + \logg k)$.
\end{proof}

\begin{proof}[Proof of  Proposition~\ref{prop:CCCvsKKK}.]
  For infinitely many~$n$ we construct strings $x_i$ of length~$n$ such that
  \begin{enumerate}
    \item The values $\K(x_i)$ are dense in an interval of size $\Omega(\logg n)$, while all values $\C(x)$
      are contained in an interval of size $O(\loggg n)$.
    \item The values $\KKK(x_i)$ are dense in an interval of size at least $\Omega(\logiv n)$, 
      while all values $\CCC(x_i)$ are contained in an interval of size $O(\logv n)$.
    \item $CCCC\,(x_i) \le O(\logv n)$. \label{item:smallCCCC}
  \end{enumerate}
  $2$ and $3$ imply Proposition~\ref{prop:CCCvsKKK}. By Lemma~\ref{lem:dense} we can already
  observe that $1$ implies $2$, thus it remains to show $1$ and $3$.

  \medskip
  We start the construction by identifying two strings $y$ and $z$ of length $n$ such that $\K(y) - \K(z) \ge \logg n$
  and $\C(y) = \C(z) = n$ (here and below we omit $O(1)$ terms). 
  More specifically our construction implies $\K(y) = \K(n) + n - \logg n$ and $\K(z) = \K(n) + n$. 
  We use the construction of~\cite[Corollary 6]{BauwensCompcomp} (which slightly differs from 
  the proof of Theorem~\ref{th:SolovayII}). Let us repeat this construction.
  As usual, let $n$ be such that $\K(\K(n)|n) \ge \logg n$.
  For the proof of item~\ref{item:smallCCCC}, note that $n$ exist for all values of $\logg n$, 
  and we choose such values that satisfy
  \begin{equation}\label{eq:Cloggn}
    \C(\logg n)  \le O(\logv n)
  \end{equation}
  (there exist infinitely many such $n$).
  Let $z$ of length $n$ be such that $\C(z|\K(n),n) \ge n$. Let $y$ be the concatenation 
  of $\K(n)$ in binary and the last $n - \logg n$ bits of $z$. By Lemma~\ref{lem:GacsTight}, 
  the length of $\K(n)$ in binary is $\logg n$, thus $|y| = n$.
  The string $y$ is the same constructed string as 
  in the proof of~\cite[Corollary 6]{BauwensCompcomp}, 
  and there it is shown using symmetry of information that
  $\C(y) = n$ and $\K(y) = \K(n) + n - \logg n$.

  What happens if for some $i \le \logg n$ in this construction 
  only the last $i$ bits from $\K(n)$ and the first $n-i$ bits from $z$  are chosen? 
  Let $x_i$ be the string obtained in this way. Note that $x_{i+1}$ is obtained 
  from $x_i$ by removing the last bit and prepending the $i+1$-th last bit of $\K(n)$.
  This implies that $\K(x_i) = \K(x_{i+1}) + O(1)$.
  For $i = 0$ we have $x_i = y$ and $\K(x_i) = \K(n) + n - \logg n$, and
  for $i = \logg n$ we have $x_i = z$ and thus $\K(x_i) = \K(n) + n$.
  This implies that the values of $\K(x_i)$ are $O(1)$-dense in an interval of size $\logg n$.
  Using symmetry of information in a similar way as before, one can show that
  $\C(x_i) = n + O(\log i)$ (we use that any $i$-bit segment of $\K(n)$ is $O(\log i)$
  incompressible given $i$).
  Recall that $i \le \logg n$, thus this implies that all values $\C(x_i)$ are contained 
  in an interval of size $O(\loggg n)$, and this finishes the proof of item $1$.

  \medskip
  We show item~\ref{item:smallCCCC}.
  Recall that  $\C(x_i) = n + O(\loggg n)$, so we need to show that 
  $\CCC(n) \le O(\logv n)$. We know that $C(\logg n) \le O(\logv n)$ but unfortunately, 
  $\CC(n)$ can contain much more information than $\logg n$. We take another approach by showing
  that $\CC(\K(n)) \le O(\logv n)$ and $\CCC(n) \le O(\CC\K(n))$.
  The last inequality follows from footnote~\ref{foot:DoubleComplexities}.
  For the first, note from footnote~\ref{foot:Gacs} that $\C(\K(n)) = \logg n$, 
  and by \eqref{eq:Cloggn} this implies $\CC(\K(n)) \le O(\logv n)$.
\end{proof}

\end{document}